\documentclass[12pt]{amsart}
\usepackage{amsmath,amscd,amsfonts,amssymb, color,bbm, bm,
braket, mathtools}
\usepackage{latexsym, graphicx, pstricks,rotating, enumerate}
\usepackage[margin=1in]{geometry}
\usepackage[colorlinks = true, linkcolor = blue, urlcolor  = blue, citecolor = red]{hyperref}
\usepackage[capitalise]{cleveref}
\definecolor{darkblue}{rgb}{0.0,0.0,0.3}
 \hypersetup{colorlinks=true,citecolor=darkblue,
 unicode=true,pdftitle=markovian-inf.pdf,pdfauthor=Ritabrata
 Sengupta, urlcolor=darkblue} 
\numberwithin{equation}{section}
\newtheorem{proposition}{Proposition}[section]
\newtheorem{definition}{Definition}[section]
\newtheorem{remark}{Remark}[section]

\newtheorem{theorem}{Theorem}[section]
\newtheorem{lemma}{Lemma}[section]
\newtheorem{corollary}{Corollary}[section]

\usepackage{mathrsfs}

\renewcommand{\epsilon}{\varepsilon}

\begin{document}

\author{Saikat Patra}
\address{Department of Mathematical Sciences, Indian Institute of
Science Education \& Research (IISER) Berhampur, Transit campus,
Berhampur, 760 010, Ganjam, Odisha, India.}
\email[Saikat Patra]{\href{mailto:saikatp@iiserbpr.ac.in}{saikatp@iiserbpr.ac.in}}

\title{A study on Heisenberg-Weyl linear maps}
\author{Bihalan Bhattacharya}
\address{Institute of Physics, Faculty of Physics, Astronomy and Informatics, Nicolaus Copernicus University, Grudziadzka 5/7, 87-100 Torun, Poland}
\email[Bihalan Bhattacharya]{\href{mailto:bihalan@umk.pl }{bihalan@umk.pl }}

\begin{abstract}
 
\par Heisenberg-Weyl operators provide a Hermitian generalization of Pauli operators in higher dimensions. Positive maps arising from Heisenberg-Weyl operators have been studied along with several algebraic and spectral properties of Heisenberg-Weyl observables. This allows to generalize the study of Pauli type maps in higher dimesional algebra of operators.    \\ 


\end{abstract}
\maketitle
	
\section{Introduction}\label{sec:1}
Study of positive linear maps between algebra of operators is interesting from various aspects \cite{book2, book1}. Let $B\left(\mathbf{C}^d\right)$ stands for linear operators acting on d dimensional complex Hilbert space $\mathbf{C}^d$. A linear map $\Phi$ acting on $B\left(\mathbf{C}^d\right)$ is said to be positive if $\Phi (X) \geq 0$ for any $X \geq 0$ i.e. it takes a positive semidefinite operator to a positive semidefinite operator and k-positive if the map $\mathbf{I}_k \otimes \Phi$ acting on $\mathbf{M}_k \otimes B\left(\mathbf{C}^d\right) $ is positive for some natural number k where $\mathbf{M}_k$ stands for the algebra of k $\times$ k complex matrices. Completely positive maps \cite{Stinespring55} are those positive maps which are k-positive for all $k \in \mathbf{N}$. Completely positive trace preserving $\left( i.e.\text{Tr}(\Phi[X])=\text{Tr}[X] \right)$ linear maps are considered as quantum channels which are often used to study the dynamics of open quantum systems \cite{C2022, breuer,rivas1,breuerN,alonso} and quantum communications \cite{NC10,W13,W18}.
Positive but not completely positive maps play instrumental role in detecting quantum entanglement\cite{P96,H96,H97,H09,GT09,CS2014}.\\

Despite of considerable efforts from both mathematics and physics community \cite{S63,C75,WORONO76,Tang86,OSAKA1991,Osaka93,Ha98,Ha02,K03,Ha03,M07,ZC13,C14,LACRS16,MR17} the structure of positive map is still uncharacterized fully even in low dimensions and it requires extensive research. 
There is no universal known procedure for construction of positive maps. 
For example in the simplest case if one considers linear maps acting on algebra of 2 $\times$ 2 complex matrices, one can construct a class of positive maps arising from Pauli matrices which is famously known as Pauli maps. Pauli operators form a orthogonal basis for the real vector space of Hermitian operators in $B\left(\mathbf{C}^2\right) $. Pauli operators are Hermitian as well as unitary. One way to construct the Hermitian generalization of Pauli operators in higher dimension is considering Generalized Gell-Mann matrices. In \cite{AEHK16}, authors have  constructed an orthogonal basis for algebra of operators acting on d-dimensional complex Hilbert space. This basis comprises of Hermitian operators, therefore 
 they are observables acting on d-level systems. These operators are reffered as Heisenberg-Weyl (H.W.) observables in \cite{AEHK16}. H.W.-operators arise from the displacement operators acting on d-dimensional systems. The aim of this construction is to generalize the Pauli operators in arbitrary dimensions. The H.W.-operators are different from generalized Gell-Mann operators in various aspects which potentially allows one to use them in different quantum information-theoretic tasks.\\

In this paper we aim to study some algebraic and spectral properties of the Heisenberg-Weyl operators  and mainly we are interested in studying positive maps which can be constructed out of these operators . Heisenberg-Weyl operators being a generalization of Pauli operators, allow us to study generalization of Pauli maps.  \\

This paper is organized as following. In section (\ref{II}) we discuss on Weyl operators and construction of generalized Pauli channels. In section (\ref{III}) we study the algebraic and spectral properties of Heisenberg-Weyl observables. We introduce Heisenberg-Weyl maps in section(\ref{IV}) and discuss its unitality condition. We further discuss construction of generalized Pauli channels with the help of H.W. observables and the possibility of further generalizations. We study the eigenvalue equation corresponding to linear maps arising from  Heisenberg-Weyl observables and  we also do a case study of Heisenberg-Weyl maps acting on $B\left(\mathbf{C}^3\right)$. Finally we conclude in section (\ref{V}).

\section{Weyl operators and Generalized Pauli channels}\label{II}
Consider d dimensional complex Hilbert space $\mathbf{C}^d$ with orthogonal basis $\lbrace \vert 0 \rangle, \vert 1 \rangle,..., \vert d-1 \rangle \rbrace $. Generalized Pauli "Shift" and "Phase" operators viz. $X$ and $Z$ are defined on $\mathbf{C}^d$ by

\begin{eqnarray*}
    X \vert r \rangle = \vert r+1  \rangle
\end{eqnarray*}
\begin{eqnarray*}
    Z \vert r \rangle = exp\left(\frac{i 2 \pi r }{d}\right) \vert r \rangle \qquad \text{with} \quad i=\sqrt{-1}
\end{eqnarray*}
Weyl operators acting on d-dimensional systems are defined by 
\begin{eqnarray*}
    W_{k,l}=X^lZ^k
\end{eqnarray*}
Note that in the definition of the X operator, the symbol $'+'$ stands for addition modulo d. \\

One can construct displacement operators for d-dimension as 
\begin{equation}\label{Weyl}
    D_{k,l}=exp\left(\frac{-i \pi k l  }{d}\right) Z^{k} X^{l}= exp\left(\frac{i \pi k l  }{d}\right) W_{k,l}.
\end{equation}
We also note here that the multiplication of k and l in the definition of $D_{k,l}$ is performed under congruence modulo d.\\

Weyl operators are unitary by construction and satisfy the following relations 
\begin{eqnarray*}
W_{k,l} W_{r,s} = exp\left( \frac{i 2 \pi ks }{d}\right) W_{k+r, l+s} ~~\text{and}~~W^{\dagger}_{k,l}=exp\left( \frac{i 2 \pi kl }{d}\right) W_{-k,-l}.
\end{eqnarray*}

Moreover Weyl operators have another interesting property of commutation. In fact given dimension d  is prime then $d^2-1$ Weyl operators excluding identity can be split into $d+1$ subsets each consisting of $d-1$ mutually commuting operators. The operators
\[\lbrace  W_{rk,rl}, r=1, 2,...,d-1 \rbrace\] with $rk=k + k +....+ k ~~  \text{ r times and}~~ rl= l + l + ..... +l $ (r times) belong to the same commuting subset. \\

 Now one can define a linear map $\Lambda_W : B\left(\mathbf{C}^d\right) \rightarrow B\left(\mathbf{C}^d\right)$ as                                                                                                                                                                                      \begin{eqnarray*}
    \Lambda_W (Y) = \sum_{k,l=0}^{d-1} p_{k,l}~ W_{k,l} Y~ W^{\dagger}_{k,l} ~~\text{with} ~~ \sum^{d-1}_{k,l} p_{k,l} = 1.
\end{eqnarray*}

The above map $\Lambda_W$ is a convex combination of Weyl operators which are unitary in nature. It is a complete positive map and is known as Weyl channels. For $d=2$ this reduces to the famous Pauli channels. Unlike Pauli channels Weyl channels are not Hermitian for any $d > 2$. To retain Hermitian property for higher dimensions the idea of generalized Pauli channels was introduced in \cite{CS2016} 
Let us briefly discuss on generalized Pauli channels here. 

  Assuming that the d-dimensional Hilbert space has the maximal number $d+1$ of MUBs, let them be $\lbrace  \vert \eta^{\alpha}_0 \rangle,\vert \eta^{\alpha}_1 \rangle, ....,, \vert \eta^{\alpha}_{d-1} \rangle \rbrace$ with $\alpha=1, 2,,,,,d+1$. One can define $d+1$ completely positive maps 
\begin{eqnarray*}
    \mathbf{V}_\alpha = d \Phi_\alpha - \mathcal{I}
\end{eqnarray*}
where $\mathbf{I}$ stands for d-dimensional identity operator 
with \begin{equation} \label{proj}
    \Phi_\alpha (\sigma)=\sum_{r=0}^{d-1} P_r^{\alpha} \sigma  P_r^{\alpha}
\end{equation}
where $P_r^{\alpha} = \vert \eta^{\alpha}_r \rangle\langle \eta^{\alpha}_r \vert $ is the rank-1 projection corresponding to the vector $\vert \eta^{\alpha}_{r} \rangle$.  Generalized Pauli channels can be defined as follows
\begin{equation}\label{GPauli}
    \Lambda = p_0 \mathcal{I}+\frac{1}{d-1} \sum_{\alpha=1}^{d+1} p_{\alpha} \mathbf{V}_\alpha
\end{equation}
Here $\mathcal{I}$ stands for the identity channel.\\

One can observe that if $\alpha \neq \beta$, then
\begin{eqnarray*}
   \Phi_{\alpha} \circ \Phi_{\beta} (\sigma) = \Phi_{\beta} \circ \Phi_{\alpha} (\sigma) = \frac{1}{d} \mathbf{I}~ \text{Tr} (\sigma)
\end{eqnarray*}

 and for a d-dimensional system, where d is prime, generalized Pauli channels are convex mixture of $d+1$ completely positive maps which mutually commute with each other.\\

It is to be noted that the construction of generalized Pauli channel in (\ref{GPauli}) relies on the existence of mutually unbiased bases (MUB). In this context it is important to mention that the Weyl operators are connected with MUB. In fact if the dimension of a system d is prime, then one can construct a complete set of MUBs with Weyl operators. In this case the eigenbases of the operators $W_{0,1},W_{1,0},W_{1,1},.....,W_{1,d-1} $ give rise to $d+1$ MUB. Therefore for prime d, one can construct generalized Pauli channels with the help of Weyl operators as well. Now our aim is to provide a different representation of generalized Pauli channels and generalize it further.

\section{Heisenberg-Weyl observables: }\label{III}
  Let us consider $D_{k,l}$,  the displacement operators acting on d-dimensional system. \\
  
  \begin{definition} 
      Heisenberg-Weyl operators \cite{AEHK16} acting on d-dimensional system is defined as

      \begin{eqnarray*}
    Q_{k,l}= \chi D_{k,l} + \chi^{*} D^{\dagger}_{k,l}
\end{eqnarray*}
where the possible values of $\chi$ are $\frac{1\pm i}{2}$,
      
\end{definition}
  
 The values of $\chi$ are obtained from the orthogonality conditions. These operators $Q_{k,l}$ are called Heisenberg-Weyl (H.W.) observables as they are Hermitian by construction and the corresponding basis is known as the Heisenberg-Weyl (H.W.) basis. In this work we have studied H.W. observables and their properties from the perspective of positive maps. Following are some properties of H.W. observables.\\

\subsection{Algebraic properties:}
In this section we shall discuss several algebraic properties of H.W. observables. These include commutation relations of H.W. observables and several identities involving them. \\
\begin{proposition}
    If H.W. observables are defined on prime dimension d, given any k,l 
    \begin{eqnarray*}
        Q_{n_1k,n_1l}~ Q_{n_2k,n_2l}=Q_{n_2k,n_2l}~Q_{n_1k,n_1l}
    \end{eqnarray*}
    where $n_ik=k + ....+k$ ($n_i$ times) and $n_il= l+....+ l$ ($n_i$ times) with $i=1,2$.
\end{proposition}
\begin{proof}
    Using the definition of $Q_{k,l}$, we note that
    \begin{eqnarray*}
        \begin{split}
         &Q_{n_1k,n_1l}~ Q_{n_2k,n_2l}\\
         &=\left( \kappa D_{n_1k,n_1l}+\kappa^* D^{\dagger}_{n_1k,n_1l}\right) \left( \kappa D_{n_2k,n_2l}+\kappa^* D^{\dagger}_{n_2k,n_2l}\right) \\
         &=\kappa^2 D_{n_1k,n_1l} D_{n_2k,n_2l} +\kappa \kappa^* D_{n_1k,n_1l}D^{\dagger}_{n_2k,n_2l}+\kappa^*\kappa D^{\dagger}_{n_1k,n_1l}D_{n_2k,n_2l}+\kappa^{*2} D^{\dagger}_{n_1k,n_1l}D^{\dagger}_{n_2k,n_2l}\\
         &= \eta \left(\kappa^2 W_{n_1k,n_1l} W_{n_2k,n_2l} +\kappa \kappa^* W_{n_1k,n_1l}W_{-n_2k,-n_2l}+\kappa^*\kappa W_{-n_1k,-n_1l}W_{n_2k,n_2l}+\kappa^{*2} W_{-n_1k,-n_1l}W_{-n_2k,-n_2l} \right)\\
         &=\eta \left(\kappa^2 W_{n_2k,n_2l} W_{n_1k,n_1l}  +\kappa \kappa^* W_{-n_2k,-n_2l} W_{n_1k,n_1l}+\kappa^*\kappa W_{n_2k,n_2l} W_{-n_1k,-n_1l}+\kappa^{*2} W_{-n_2k,-n_2l} W_{-n_1k,-n_1l} \right)\\
         &=\kappa^2D_{n_2k,n_2l}  D_{n_1k,n_1l}  +\kappa \kappa^* D_{n_2k,n_2l} D^{\dagger}_{n_1k,n_1l}+\kappa^*\kappa D^{\dagger}_{n_2k,n_2l}D_{n_1k,n_1l}+\kappa^{*2} D^{\dagger}_{n_2k,n_2l} D^{\dagger}_{n_1k,n_1l}\\
         &=\left( \kappa D_{n_2k,n_2l}+ \kappa^* D^{\dagger}_{n_2k,n_2l}\right) \left(\kappa D_{n_1k,n_1l}+ \kappa^* D^{\dagger}_{n_1k,n_1l} \right)\\
         &=Q_{n_2k,n_2l}~Q_{n_1k,n_1l}
        \end{split}
    \end{eqnarray*}
    where we have used the commutation property of Weyl operators in prime dimension \cite{CS2016} along with $\eta=exp \left( \frac{\pi i (n_1k~n_1l+n_2k~n_2l)}{d}\right)$.
\end{proof}
Above proposition leads to the following remark.\\

\textbf{Remark:} If dimension of the system is prime and let us consider the set of operators $\lbrace Q_{nk,nl}: n=1, 2,...., d-1 \rbrace$ with arbitrary k and l where $nk=k +...+ k~~ (n \text{times}),nl=l +....+ l (n~ \text{times})$. The elements in this set commute with each other. \\

\begin{proposition} \label{p1} Given any $d \in \mathbf{N}$, \\
   \begin{center}
 $ Q_{k,l}^2 + Q_{d-k,d-l}^2 = 2 \mathbf{I}$
\end{center}
\end{proposition}  
    
\begin{proof}
    Let us recall that the operator $Q_{kl}=\chi D_{kl}+\chi^{*} D^{\dagger}_{kl}$, where $\chi=\frac{1+i}{2} \text{or} \frac{1-i}{2}$ and $D_{kl}=exp\left(\frac{-i \pi k l }{d}\right) Z^{k} X^{l}$. The operator Z acts on $\mathbf{C}^d$ as $Z \vert v \rangle = exp\left(\frac{2i \pi v}{d} \right)\vert v \rangle$ and the operator X  acts on $\mathbf{C}^d$ as $X \vert v \rangle = \vert v+1 \rangle.$ Therefore the operator $Z^{k}X^{l}$ acts on $\mathbf{C}^d$ as 
\begin{eqnarray*}
    Z^{k}X^{l} \vert v \rangle = exp \left( \frac{2 \pi i k (v+l)}{d}\right) \vert v + l \rangle
\end{eqnarray*}
and its adjoint as 
\begin{eqnarray*}
    \left(Z^{k}X^{l}\right)^{\dagger} \vert v \rangle = exp \left( \frac{-2 \pi i k v}{d}\right) \vert v - l \rangle
\end{eqnarray*}
Therefore,
Now it is enough to prove that, $\left(D^2_{kl}+D^2_{d-k,d-l}-D^{\dagger2}_{kl}-D^{\dagger2}_{d-k,d-l}\right)$ is zero operator.\\
Note that, 
\begin{eqnarray*}
    \begin{split}
        &\left(D^2_{kl}+D^2_{d-k,d-l}-D^{\dagger2}_{kl}-D^{\dagger2}_{d-k,d-l}\right)\\
        &=\begin{aligned}[t]
        &\left(exp \left( \frac{-4 \pi i k l}{d}\right) Z^{2k}X^{2l} -exp \left( \frac{4 \pi i k l}{d}\right) (Z^{2k}X^{2l})^{\dagger}\right)\\
        &+\left(exp \left( \frac{-4 \pi i(d-k)(d-l)}{d}\right) Z^{2(d-k)}X^{2(d-l)}-exp \left( \frac{4 \pi i (d-k)(d-l)}{d}\right) (Z^{2(d-k)}X^{2(d-l)})^{\dagger}\right)
        \end{aligned}\\
        &=\begin{aligned}[t]
        &exp \left( \frac{4 \pi ik(v+l)}{d}\right) \vert v+2 l \rangle-exp \left( \frac{4 \pi ik(l-v)}{d}\right) \vert v-2 l  \rangle \\
        &+exp \left( \frac{4 \pi i(d-k)(d+v-l)}{d}\right) \vert v-2 l+2d   \rangle-exp \left( \frac{4 \pi i(d-k)(d-v-l)}{d}\right) \vert v+2 l-2d  \rangle\end{aligned}\\
        &=\begin{aligned}[t]
            &\left(exp \left( \frac{4 \pi ik(v+l)}{d}\right)-exp \left( \frac{4 \pi i(d-k)(d-v-l)}{d}\right)\right)\vert v+2 l \rangle\\
            &-\left(exp \left( \frac{4 \pi ik(l-v)}{d}\right)-exp \left( \frac{4 \pi i(d-k)(d+v-l)}{d}\right) \right)\vert v-2 l  \rangle 
        \end{aligned}\\
        &= \vert 0 \rangle
        \end{split}
\end{eqnarray*}

Therefore, $Q^{2}_{kl}+Q^{2}_{d-k,d-l}= 2 \mathbf{I} $  \\

\end{proof}

\begin{proposition} For any $d \in \mathbf{N}$ 
    $$    \sum_{n=1}^{d-1}  Q^{2}_{nk,nl} = (d-1) \mathbf{I} $$
\end{proposition}

\begin{proof}
From the definition of the operators $Q_{k,l}$, given  left hand side can be expressed as, \[\sum_{n=1}^{d-1} Q^2_{nk,nl}=\sum_{n=1}^{d-1}\chi^2\left(D^2_{nk,nl}-D^{\dagger 2}_{nk,nk}\right)+(d-1)\mathbf{I}\]
 as $\sum_{n=1}^{d-1}\chi^2\left(D^2_{nk,nl}-D^{\dagger 2}_{nk,nk}\right)$ is zero operator, therefore 
 \[\sum_{n=1}^{d-1} Q^2_{nk,nl}=(d-1)\mathbf{I}\]
\end{proof}

\subsection{Spectral properties:}

\begin{definition}
    Two operators are said to be isospectral if they share a common spectrum.
\end{definition}

\begin{proposition} For any k,l the traceless operators $(-1)^{k \ell} Q_{k\ell}$ are isospectral.

\end{proposition}
Before proving the proposition let us prove the following lemma first.\\

\begin{lemma} If $d$ is even
\begin{eqnarray*}
    (Z^k X^\ell)^d = (-1)^{k \ell} \mathbf{I} ,
\end{eqnarray*}
and if $d$ is odd
\begin{eqnarray*}
    (Z^k X^\ell)^d =  \mathbf{I} ,
\end{eqnarray*}

\end{lemma}
\begin{proof}
Let us note that,\\
\begin{eqnarray*} 
        \begin{split}
(Z^k X^\ell)^d
    &=(Z^k X^\ell)(Z^k X^\ell)...(\text{d times})...(Z^k X^\ell)
\\[1ex]
 &=Z^k\left((X^lZ^k)(X^lZ^k)...(\text{d-1 times})..(X^lZ^k)\right)X^l
\\[1ex]
&=
 exp \left( \frac{-2 \pi ikl (d-1)}{d}\right) Z^k \left((Z^kX^l)(Z^kX^l)....(\text{d-1 times})...(Z^kX^l)\right)X^l\\[1ex]
&= exp \left( \frac{-2 \pi ikl (d-1)}{d}\right) Z^{2k} \left( (X^lZ^k)(X^lZ^k)...(\text{d-2 times})...(X^lZ^k)\right)X^{2l}\\[1ex]
&=exp \left( \frac{-2 \pi ikl (d-1)}{d}\right)exp \left( \frac{-2 \pi ikl (d-2)}{d}\right)Z^{2k}\left((Z^kX^l)(Z^kX^l)...(\text{d-2 times})...(Z^kX^l) \right)X^{2l}\\[1ex]
&=exp \left( \frac{-2 \pi ikl \left((d-1)+(d-2)\right)}{d}\right) Z^{3k} \left( (X^lZ^k)(X^lZ^k)...(\text{d-3 times})...(X^lZ^k)\right)X^{3l}
\end{split}
\end{eqnarray*}
Thus using the relationship $Z^kX^l=exp \left( \frac{2 \pi ikl}{d}\right)X^lZ^k$ repeatedly we obtain\\

\begin{eqnarray*}\begin{split}
    \left(Z^kX^l \right)^d=&exp \left( \frac{-2 \pi ikl \left((d-1)+(d-2)+...+1 \right)}{d}\right) Z^{dk}X^{dl}\\[1ex]
  &=exp \left( \frac{-2 \pi ikl ~~d(d-1)}{2d}\right)\mathbf{I}\\[1ex]
  &=exp \left( - \pi i (d-1) k l\right)\mathbf{I}
    \end{split}
\end{eqnarray*}

Therefore

\begin{eqnarray*}
    (Z^k X^\ell)^d = (-1)^{k \ell} \mathbf{I} ,
\end{eqnarray*} if $d$ is even
and 
\begin{eqnarray*}
    (Z^k X^\ell)^d = \mathbf{I} ,
\end{eqnarray*}
if $d$ is odd.


\end{proof}

\vspace{1cm}

\begin{corollary} Let us denote $\omega=exp \left( \frac{2\pi i}{d}\right)$. For any k,l the spectrum of $(-1)^{k\ell} D_{k\ell}$  are given by,

\begin{eqnarray*}
    \sigma((-1)^{k\ell} D_{k\ell}) = \{1,\omega,\ldots,\omega^{d-1} \} \ , 
\end{eqnarray*}
that is, $(-1)^{k\ell} D_{k\ell}$ are isospectral.    
\end{corollary}

\begin{corollary} Since $(-1)^{k\ell} D_{k\ell}$ is normal (being unitary), one concludes that $(-1)^{k\ell} Q_{k\ell}$
are isospectral.
\end{corollary}
 \begin{proposition} For any prime d, all operators

$$    \sum_{n=1}^{d-1} (-1)^{k \ell} Q_{nk,n\ell}$$
are isospectral.
\end{proposition}
\begin{proof}
    If d is prime then the operators $\lbrace Q_{nk,nl}: n=1,2,...,d-1 \rbrace$ commute with each other. Therefore the result follows.
\end{proof}

\section{ Heisenberg-Weyl maps  } \label{IV}
In this section we discuss linear maps arising from H.W. observables. Pauli maps specially Pauli channels on $B\left(\mathbf{C}^2\right)$ are very well studied and H.W. observables provide a generalization of Pauli operators. Let us discuss unitality condition for H.W. maps acting on d-dimensional systems.  \\

Let us consider a linear map $\Lambda : B\left(\mathbf{C}^d\right) \rightarrow B\left(\mathbf{C}^d\right)$ defined as 
\begin{equation}\label{H.W. map}
    \Lambda (Y) = \sum_{k,l=0}^{d-1} p_{k,l}~ Q_{k,l} Y~ Q_{k,l}
\end{equation}
\begin{proposition}\label{unital}
    Given any $d \in \mathbf{N}$ if $p_{k_1,l_1}=p_{k_2,l_2}$ where $k1\oplus k2= 0~~ \left(\text{mod}~~ d\right)$ and  $l1\oplus l2=0 ~~\left(\text{mod}~~ d\right)$ and $p_{0,0}=1-\left( \sum_{k,l: (k,l)\neq (0,0)} p_{k,l} \right)$, then the linear map $\Lambda$ is unital. \\
\end{proposition}

\begin{proof}
    Given the definition of $\Lambda$, one has
    \[\Lambda(\mathbf{I})=\sum_{k,l=0}^{d-1} p_{k,l}~ Q^2_{k,l} \]
    \textbf{Case 1: d is odd;}
    \[\Lambda(\mathbf{I})=p_{0,0} \mathbf{I}+\sum_{k,l} p_{k,l}~ Q^2_{k,l} ~~\text{where}~~ Q_{0,0}=\mathbf{I} \]
    where the summation contains $d^2-1$ terms. These terms can be grouped into $\frac{d^2-1}{2}$ pairs by identifying \[p_{k_1,l_1}=p_{k_2,l_2} ~~\text{where}~~ k1\oplus k2= 0~~ \left(\text{mod}~~ d\right)~~ \text{and}~~  l1\oplus l2=0 ~~\left(\text{mod}~~ d\right)\].
    Now we use the proposition \ref{p1} which tells $   Q_{k,l}^2 + Q_{k',l'}^2 = 2 \mathbf{I} . $ with $k\oplus k'= 0~~ \left(\text{mod}~~ d\right)$ and  $l\oplus l'=0 ~~\left(\text{mod}~~ d\right)$. Finally identifying $p_{0,0}=1-\left( \sum_{k,l: (k,l)\neq (0,0)} p_{k,l} \right)$, we obtain $\Lambda(\mathbf{I})=\mathbf{I}$\\

\textbf{Case 2: d is even;}

When d is even, we write 
\[\Lambda(\mathbf{I})=\sum_{k,l=0}^{d-1} p_{k,l}~ Q^2_{k,l}=p_{0,0} \mathbf{I} + p_{0,1} Q^2_{0,1}+...+p_{0,\frac{d}{2}} Q^2_{0,\frac{d}{2}}+....+p_{0,d-1} Q^2_{0,d-1} \]
\[+p_{1,0} Q^2_{1,0}+...+p_{\frac{d}{2},0} Q^2_{\frac{d}{2},0}+...+p_{d-1,0} Q^2_{d-1,0}+ \sum_{k,l} p_{k,l}~ Q^2_{k,l}\]
    Note that \[Q^2_{0,\frac{d}{2}}=Q^2_{\frac{d}{2},0}=Q^2_{\frac{d}{2},\frac{d}{2}}=\mathbf{I}\].

    Then identifying $p_{k,l}=p_{k',l'}~~ \text{where}~~k\oplus k'= 0~~ \left(\text{mod}~~ d\right)$ and  $l\oplus l'=0 ~~\left(\text{mod}~~ d\right) ~~ \text{and}~~ p_{0,0}=1-\left( \sum_{k,l: (k,l)\neq (0,0)} p_{k,l} \right)$ we get $\Lambda(\mathbf{I})=\mathbf{I}$.
\end{proof}
\vspace{1cm}

\textbf{Remark:} The above conditions for unitality of the map $\Lambda$ is sufficient only, not necessary one.\\

For example let us consider $d=4$ and the map
\begin{equation} \label{H.W.4}
    \Lambda(Y)=\sum^3_{k,l=0} p_{k,l}~Q_{k,l} Y Q_{k,l}
\end{equation}

It can be shown that 
\[\Lambda(\mathbf{I})=\left( \sum_{k,l} p_{k,l} \right) \mathbf{I}\]

Therefore, putting $p_{0,0}=1-\left( \sum_{k,l: (k,l)\neq (0,0)} p_{k,l} \right)$ we obtain unitality but it does not necessarily mean $p_{k,l}=p_{k',l'}~~ \text{where}~~k\oplus k'= 0~~ \left(\text{mod}~~ d\right)$ and $   l\oplus l'=0 ~~\left(\text{mod}~~ d\right)$. In fact $p_{0,0}=1-\left( \sum_{k,l: (k,l)\neq (0,0)} p_{k,l} \right)$  is the necessary and sufficient condition for unitality of the map defined on (\ref{H.W.4}).\\

\textbf{Remark:} The above unitality condition allows further refinement in a set of mutually commuting H.W. observables. If dimension of the system is d is prime then each set of mutually commuting H.W. observables consists of $d-1$ operators. Let us consider a subset consisting of $\lbrace Q_{k_1,l_1}, Q_{k_2,l_2}\rbrace$ with $k_1 \oplus k_2=d$ and $l_1 \oplus l_2 =d$. Since the operations here are being carried out under congruence modulo d, we can rewrite the same subset in the form $\lbrace Q_{k,l}, Q_{-k,-l}\rbrace$. 
Then for each set of commuting H.W. observables we have $\frac{d-1}{2}$ such subsets.\\

\begin{theorem}\label{theorem1}
    Given any $d \in \mathbb{N} $,  let us consider two linear maps $\Psi_1: B\left(\mathbf{C}^d\right) \rightarrow B\left(\mathbf{C}^d\right) $ and $\Psi_2:B\left(\mathbf{C}^d\right) \rightarrow B\left(\mathbf{C}^d\right)$ given by
        \[\Psi_1 (Y)=Q_{k_1,l_1} Y Q_{k_1,l_1} + Q_{-k_1,-l_1} Y Q_{-k_1,-l_1} ~~\text{and}~~\Psi_2 (Y)=Q_{k_2,l_2} Y Q_{k_2,l_2} + Q_{-k_2,-l_2} Y Q_{-k_2,-l_2}\]
    Then
    \begin{center}
        $\Psi_1 \circ \Psi_2 = \Psi_2 \circ \Psi_1$
    \end{center}
\end{theorem}

\vspace{0.3cm}

\begin{proof}

Given any $d \in \mathbb{N}$, two linear maps $\Psi_1$ and $\Psi_2$ defined on $B\left(\mathbf{C}^d\right)$ are given by
    \[\Psi_1 (Y)=Q_{k_1,l_1} Y Q_{k_1,l_1} + Q_{-k_1,-l_1} Y Q_{-k_1,-l_1} ~~\text{and}~~\Psi_2 (Y)=Q_{k_2,l_2} Y Q_{k_2,l_2} + Q_{-k_2,-l_2} Y Q_{-k_2,-l_2}\]

Now let $Y \in B\left(\mathbf{C}^d\right)$ be an arbitrary operator. Then,

    \begin{equation} \label{W}
        \begin{split}
\Psi_1 \circ \Psi_2  (Y)&=\Psi_1 \left( \Psi_2 (Y)\right)
\\[1ex]
 &=\Psi_1 \left( Q_{k_2,l_2} Y Q_{k_2,l_2} + Q_{-k_2,-l_2} Y Q_{-k_2,-l_2} \right)
\\[1ex]
&=\begin{aligned}[t] 
 &Q_{k_1,l_1}\left(Q_{k_2,l_2} Y Q_{k_2,l_2} + Q_{-k_2,-l_2} Y Q_{-k_2,-l_2}  \right)Q_{k_1,l_1} + \\
& Q_{-k_1,-l_1} \left( Q_{k_2,l_2} Y Q_{k_2,l_2} + Q_{-k_2,-l_2} Y Q_{-k_2,-l_2}\right) Q_{-k_1,-l_1}
\end{aligned}
\\[1ex]
&=\begin{aligned}[t]
&Q_{k_1,l_1} Q_{k_2,l_2} Y \left(Q_{k_1,l_1} Q_{k_2,l_2} \right)^{\dagger}+ Q_{k_1,l_1} Q_{-k_2,-l_2} Y \left(Q_{k_1,l_1} Q_{-k_2,-l_2} \right)^{\dagger}+\\
 &Q_{-k_1,-l_1} Q_{k_2,l_2} Y \left(Q_{-k_1,-l_1} Q_{k_2,l_2} \right)^{\dagger} + Q_{-k_1,-l_1} Q_{-k_2,-l_2} Y \left(Q_{-k_1,-l_1} Q_{-k_2,-l_2} \right)^{\dagger}
       \end{aligned}
        \end{split}
    \end{equation}

    Similarly \\
\begin{equation}\label{W1} \begin{split}
    \Psi_2 \circ \Psi_1 (Y)=Q_{k_2,l_2} Q_{k_1,l_1} Y \left( Q_{k_2,l_2} Q_{-k_1,-l_1}\right)^{\dagger} + Q_{k_2,l_2} Q_{-k_1,-l_1} Y \left( Q_{k_2,l_2} Q_{-k_1,-l_1}\right)^{\dagger}\\[1ex]
+Q_{-k_2,-l_2} Q_{k_1,l_1} Y \left( Q_{-k_2,-l_2} Q_{k_1,l_1}\right)^{\dagger}+Q_{-k_2,-l_2} Q_{-k_1,-l_1} Y \left( Q_{-k_2,-l_2} Q_{-k_1,-l_1}\right)^{\dagger}
\end{split}
 \end{equation}
 \vspace{0.2cm}
 
    Let us recall that for any $k,l\in \mathbf{N}$,  $Q_{k,l}=\chi D_{k,l} + \chi^{*} D_{k,l}^{\dagger}$ .\\

    Now one notes that from the relationship between displacement operator and Weyl operator discussed in (\ref{Weyl}),
       
        \begin{equation}\label{dagger}
            \begin{split}
                D_{k,l}^{\dagger}&=exp \left( \frac{-\pi i kl}{d}\right) W_{k,l}^{\dagger}\\
                &=exp \left( \frac{-\pi i kl}{d}\right)exp \left( \frac{2\pi i kl}{d}\right)W_{-k,-l}\\
                &=exp \left( \frac{\pi i kl}{d}\right)W_{-k,-l}\\
                &=exp \left( \frac{\pi i (-k)(-l)}{d}\right)W_{-k,-l}\\
                &=D_{-k,-l}
            \end{split}
        \end{equation}

  and for an arbitrary operator Y and for arbitrary $k,l,k',l' $, we have,
  \begin{equation}\label{product}
      \begin{split}
&D_{k,l}D_{k',l'}~ Y~ D_{k',l'}D_{k,l}\\
&=exp \left( \frac{\pi i (kl+k'l')}{d}\right)W_{k,l}W_{k'l'}~Y~exp \left( \frac{\pi i (kl+k'l')}{d}\right)W_{k'l'}W_{k,l}\\
&=\exp \left( \frac{\pi i (kl+k'l')}{d}\right)exp \left( \frac{2\pi i kl'}{d}\right) W_{k+k',l+l'} ~Y ~exp \left( \frac{\pi i (kl+k'l')}{d}\right) exp \left( \frac{2 \pi i k'l}{d}\right)W_{k'+k,l'+l}\\
&=exp \left( \frac{\pi i (kl+k'l')}{d}\right)exp \left( \frac{2 \pi i k'l}{d}\right)W_{k'+k,l'+l}~Y~ exp \left( \frac{\pi i (kl+k'l')}{d}\right)exp \left( \frac{2\pi i kl'}{d}\right)W_{k+k',l+l'}\\
&=exp \left( \frac{\pi i (kl+k'l')}{d}\right) W_{k',l'}W_{k,l}~Y~exp \left( \frac{\pi i (kl+k'l')}{d}\right)W_{k,l}W_{k',l'}\\
&=D_{k'l'}D_{k,l}Y D_{k,l}D_{k'l'}
      \end{split}
  \end{equation}

Now expanding a term of the form  $Q_{k,l} Q_{k',l'}~Y~ Q_{k',l'}Q_{k,l}$ in terms of $D_{k,l}$ and $D^{\dagger}_{k,l}=D_{-k.-l}$ we find,\\

     $Q_{k,l} Q_{k',l'}Y Q_{k',l'}Q_{k,l}$
    \begin{tiny}
    
   \begin{eqnarray*}
\begin{split}&=\left(\chi^2 D_{k,l} D_{k'l'}+\chi^{*}\chi D_{k,l}^{\dagger} D_{k'l'}+\chi \chi^{*} D_{k,l} D_{k'l'} ^{\dagger}+\chi^{* 2} D_{k,l}^{\dagger} D_{k'l'}^{\dagger}\right) Y \left( \chi^2 D_{k',l'}D_{k,l}+\chi^{*}\chi D_{k'l'}^{\dagger}D_{k,l}+\chi\chi^{*}D_{k',l'}D_{k,l}^{\dagger}+\chi^{*2}D_{k',l'}^{\dagger}D_{k,l}^{\dagger}\right)\\[1ex]
    &=\left(\chi^2 D_{k,l} D_{k'l'}+\chi^{*}\chi D_{-k,-l} D_{k'l'}+\chi \chi^{*} D_{k,l} D_{-k',-l'} +\chi^{* 2} D_{-k,-l} D_{-k',-l'}\right) Y \\[1ex] 
&\left( \chi^2 D_{k',l'}D_{k,l}+\chi^{*}\chi D_{-k',-l'}D_{k,l}+\chi\chi^{*}D_{k',l'}D_{-k,-l}+\chi^{*2}D_{-k',-l'}D_{-k,-l}\right)
\end{split}
\end{eqnarray*}
    \end{tiny}
    Therefore each of the equations (\ref{W}) and (\ref{W1}) contains a total of 64 terms.  Hence using the above relationships (\ref{dagger}) and (\ref{product}), one can establish the equality of (\ref{W}) and (\ref{W1}). 

\end{proof}

\subsection{Generalized Pauli channels}
If d is prime, then let us define $d+1$ completely positive maps as following:
\begin{eqnarray*}
    \Phi_\alpha(X)=\frac{1}{d} \sum_{n=0}^{d-1} Q_{n,\alpha n} X Q_{n,\alpha n} ~~\text{for} ~~\alpha=1,2,...,d-1
\end{eqnarray*}
\begin{eqnarray*}
    \Phi_d(X)=\frac{1}{d}\sum_{n=0}^{d-1} Q_{n,0}XQ_{n,0} ~~\text{and}~~\Phi_{d+1}(X)=\frac{1}{d}\sum_{n=0}^{d-1} Q_{0,n}XQ_{0,n}
\end{eqnarray*}
One can note that \begin{equation}\label{GP1}
    \mathbf{U}_\alpha(X)=d \Phi_{\alpha}(X)-\mathcal{I}(X)=\sum_{n=1}^{d-1} Q_{n,\alpha n}XQ_{n,\alpha n} ~~\text{for} ~~\alpha=1,2,...,d-1 
\end{equation}
\begin{equation}\label{GP2}
    \mathbf{U}_d(X)=d \Phi_d (X)-\mathcal{I}(X)=\sum_{n=1}^{d-1} Q_{n,0 }XQ_{n,0}
\end{equation} and
\begin{equation}\label{GP3}
    \mathbf{U}_{d+1}=d \Phi_{d+1}(X)-\mathcal{I}(X)=\sum_{n=1}^{d-1} Q_{0,n}XQ_{0,n}
\end{equation}
Now one can define a linear map $\Lambda: B\left(\mathbf{C}^d\right) \rightarrow B\left(\mathbf{C}^d\right) $ given by
\begin{equation}\label{GP}
    \Lambda(X)=p_0 \mathcal{I}(X)+\frac{1}{d-1}\sum_{k=1}^{d+1} p_k \mathbf{U}_k
\end{equation}
where $\left(p_0,p_1,...,p_{d+1}\right)$ is a probability vector. This is a convex mixture of total $d+2$ complete positive trace preserving maps acting on d dimensional system where d is prime and hence itself a channel.\\
\begin{proposition}
    The channel given in (\ref{GP}) corresponds to a generalized Pauli channel constructed in (\ref{GPauli}) for a prime dimension d.
\end{proposition}
     
\begin{proof}
    Let the dimension of the Hilbert space d be prime and let us consider $d+1$ mutually unbiased bases be given by
    \[\lbrace \vert \eta_0^{\alpha} \rangle,\vert \eta_1^{\alpha} \rangle,....,\vert \eta_{d-1}^{\alpha} \rangle\rbrace , \alpha=1,2,...,d+1\]

    Therefore the set 
    \[\lbrace \vert \eta_k^{\alpha} \rangle\langle \eta_l^{\alpha} \vert : k=0,1,...,d-1 ;~~l=0,....,d-1 \rbrace\] is an orthonormal basis for the algebra $B\left(\mathbf{C}^d\right)$. It is also known that the set 
    \[\lbrace\frac{1}{\sqrt{d}}Q_{k,l}: k=0,1...,d-1; ~~l=0,1,...,d-1\rbrace\] gives another orthonormal basis for $B\left(\mathbf{C}^d\right)$.\\

    We have seen that if d is prime, the set $\lbrace Q_{k,l}: k=0,1...,d-1; ~~l=0,1,...,d-1\rbrace$ can be expressed as the disjoint union of $d+1$ subsets each consisting of $d-1$ operators i.e.

    \begin{eqnarray*}
    \begin{split}
   &\lbrace Q_{k,l}: k=0,1...,d-1; ~~l=0,1,...,d-1\rbrace\\
   &=\bigsqcup_{\alpha=1}^{d-1}  \lbrace Q_{n,\alpha n} :n=0,...,d-1 \rbrace \bigsqcup \lbrace Q_{n,0} :n=0,...,d-1\rbrace\bigsqcup \lbrace Q_{0,n} :n=0,...,d-1 \rbrace 
    \end{split}
    \end{eqnarray*}
Note that, thus each of the operators $Q_{k,l}$ in any subset can be indexed by a single parameter n with $Q_{n}$ and we have,
\begin{eqnarray*}
    \frac{1}{\sqrt{d}}Q_{n}=\sum_{r} a_{n,r} P_{r} ~~\text{with}~~\sum_{n} a_{n,r} a_{n,r'}=\delta_{r,r'}
\end{eqnarray*}
Here $P_{r}=\vert \eta_r^{\alpha} \rangle\langle \eta_r^{\alpha}\vert$.\\

    Now for any $\alpha \in \lbrace 1,2,...,d+1\rbrace$, the map
    \begin{eqnarray*}
        \begin{split}
            \Phi_{\alpha}(X)=&\frac{1}{d} \sum_{n=0}^{d-1} Q_{n}~ X~ Q_{n}\\
         &=\sum_{n=0}^{d-1} \frac{1}{\sqrt{d}}Q_{n} X \frac{1}{\sqrt{d}}Q_{n}\\
         &=\sum_{n} \sum_{r} a_{n,r} P_r ~X ~\sum_{r'} a_{n,r'} P_{r'}\\
         &=\sum_{r,r'} \left( \sum_{n} a_{n,r} a_{n,r'}\right) P_r ~X~P_{r'}\\
         &=\sum_{r=0}^{d-1} P_r~X~P_{r} ~~~~\text{as} ~~\sum_{n} a_{n,r} a_{n,r'}=\delta_{r,r'}
        \end{split}
    \end{eqnarray*}
    Therefore the linear map defined in (\ref{GP}) is same as the  generalized Pauli channel constructed in (\ref{GPauli}).
\end{proof}

One can note that if d is prime then the eigenbases of the operators $Q_{0,1},Q_{1,0},Q_{1,1},.....,Q_{1,d-1}$ also give rise to MUB and hence the similar construction of generalized Pauli channel is possible with the help of H.W. observables. One advantage of using H.W. observables in this construction is that they are Hermitian in nature and hence we have another representation of the generalized Pauli channels with Hermitian Krauss operators. Interestingly,  we have a possibility to generalize this construction further. \\

Recall that if d is prime then for each $\alpha=1,2...,d-1$ the sets $\lbrace Q_{n,\alpha n}\rbrace_{n=1,2,...,d-1}$ can be further divided into $\frac{d-1}{2}$ subsets each consisting of two operators $ Q_{k,l}, Q_{-k,-l}$ for some $k,l$.  Similarly each of the sets $\lbrace Q_{n,0}\rbrace_{n=1,2,...,d-1}$ and the set $\lbrace Q_{0,n}\rbrace_{n=1,2,...,d-1}$ can also be divided into $\frac{d-1}{2}$ subsets. So total there are $\frac{d^2-1}{2}$ subsets each consists of two operators $ Q_{k,l}, Q_{-k,-l}$ for some $k,l$.\\

For each $\alpha=1,2,...,d-1,~~ \mathbf{U}_\alpha$ in (\ref{GP1}) and $\mathbf{U}_d$ in (\ref{GP2})and $\mathbf{U}_{d+1}$ in (\ref{GP3}) is a sum of $d-1$ terms. Therefore for prime d, 
one can write for $\alpha=1,2,...,d-1$,
\begin{eqnarray*}
    \mathbf{U}_\alpha(.) = \sum_{i=1}^{\frac{d-1}{2}} \mathbf{U}^{(i)}_\alpha(.) ~~\text{where}~~\mathbf{U}^{(i)}_\alpha(.)=Q_{k_i, \alpha k_i}(.)Q_{k_i, \alpha k_i}+Q_{-k_i,-\alpha k_i}(.)Q_{-k_i,-\alpha k_i}
\end{eqnarray*}
\begin{eqnarray*}
    \mathbf{U}_d(.) = \sum_{i=1}^{\frac{d-1}{2}} \mathbf{U}^{(i)}_d(.)~~\text{where}~~\mathbf{U}^{(i)}_d(.)=Q_{k_i, 0}(.)Q_{k_i, 0}+Q_{-k_i,0}(.)Q_{-k_i,0}
\end{eqnarray*} and
\begin{eqnarray*}
    \mathbf{U}_{d+1}(.) = \sum_{i=1}^{\frac{d-1}{2}} \mathbf{U}^{(i)}_{d+1}(.) ~~\text{where}~~\mathbf{U}^{(i)}_{d+1}(.)=Q_{0, k_i}(.)Q_{0, k_i}+Q_{0,-k_i}(.)Q_{0,-k_i}
\end{eqnarray*}
Therefore one can define a linear map $\Psi:B\left(\mathbf{C}^d\right) \rightarrow B\left(\mathbf{C}^d\right)$ where d is prime,
\begin{equation} \label{GGP}
    \Psi(X)=p_0 \mathcal{I} (X)+\frac{1}{d-1}\left[\sum_{\alpha=1}^{d-1} \sum_{i=1}^{\frac{d-1}{2}}p_{\alpha}^{i} \mathbf{U}_{\alpha}^{i}(X)+\sum_{i}^{\frac{d-1}{2}}p_{d}^{i}\mathbf{U}_d^{i}(X)+\sum_{i}^{\frac{d-1}{2}}p_{d+1}^{i}\mathbf{U}_{d+1}(X)\right]
\end{equation}

By virtue of the theorem (\ref{theorem1}) one can note that any two linear maps $\mathbf{U}_\gamma^{m}$ and $\mathbf{U}_\gamma^{n}$ commute for any $m,n=1,2,...,\frac{d-1}{2}$ and $\gamma=1,2,...,d+1$ i.e.
\begin{eqnarray*}
    \mathbf{U}_\gamma^{m} \circ \mathbf{U}_\gamma^{n} = \mathbf{U}_\gamma^{n} \circ \mathbf{U}_\gamma^{m}
\end{eqnarray*}
 for any $m,n=1,2,...,\frac{d-1}{2}$ and $\gamma=1,2,...,d+1$.\\

Additionally,  if \[p_{\alpha}^{i}=p_\alpha~~\text{for}~~i=1,2,...,\frac{d-1}{2}\] and \[p_{d}^{i}=p_d ~~\text{for}~~i=1,2,...,\frac{d-1}{2}\]and\[p_{d+1}^{i}=p_{d+1}~~\text{for}~~i=1,2,...,\frac{d-1}{2}\]
Then the generalized Pauli channel $\Lambda$ in (\ref{GP}) can be obtained as a special case of the linear map $\Psi$ in (\ref{GGP}).\\

\subsection{ Case study on positive maps arising from Heisenberg-Weyl operators:}


Let $\Phi:B\left(\mathbf{C}^{d_1}\right) \rightarrow B\left(\mathbf{C}^{d_2}\right) $
be a linear map and let
\[\lbrace \Gamma_0=\frac{\mathbf{I}_{d_1}}{\sqrt{d_1}}, \Gamma_i , i=1,2,....,d_1^2-1 \rbrace\] be an orthonormal basis for the Hermitian operators in $B\left(\mathbf{C}^{d_1}\right)$ and that corresponding to the Hermitian operators in $B\left(\mathbf{C}^{d_2}\right)$ be 
\[\lbrace \Omega_0=\frac{\mathbf{I}_{d_2}}{\sqrt{d_2}}, \Omega_i , i=1,2,....,d_2^2-1 \rbrace\] where $\Gamma_i (i>0) $ and $\Omega_j (j>0)$ are traceless operators. \\
Let us write the action of the linear map $\Phi$ on the basis of the domain space as
\begin{equation} \label{map}
    \Phi(\Gamma_\alpha)=\sum_{\beta=0}^{d_2^{2}-1} R_{\beta \alpha} \Omega_{\beta}
\end{equation}

Then the matrix $R=\left( R_{\beta \alpha}\right)_{\beta=0,..,d_2^{2}-1; \alpha=0,...,d_1^{2}-1} $ is given by\\
\begin{center}
    $R=\begin{bmatrix}
    R_{00}& t\\
    s& \Delta
\end{bmatrix}$
\end{center}
where $R_{00} \in \mathbf{R}$, $t \in \mathbf{R}^{d_2^{2}-1}$, $s \in \mathbf{R}^{d_1^{2}-1}$ and $\Delta$ is a $d_2^{2}-1 \times d_2^{2}-1$ real matrix.\\

Recently, a sufficient condition for positivity of a linear map defined on $B\left(\mathbf{C}^d\right)$ has been derived in \cite{Jannesary2025}.

\begin{proposition}\cite{Jannesary2025}\label{unital1}
   $\Phi$ is unital if an only if $t=0$ and $R_{00}=\sqrt{\frac{d_2}{d_1}}$ and $\Phi$ is trace preserving if and only if $s=0$ and $R_{00}=\sqrt{\frac{d_1}{d_2}}$.
\end{proposition}
\begin{proposition}\cite{Jannesary2025}{\label{positive}}
    If \begin{eqnarray*}
        \sqrt{d_2-1} \vert\vert t\vert\vert+\sqrt{d_1-1} \vert\vert s\vert\vert+\sqrt{(d_2-1)(d_1-1)}  \vert\vert \Delta\vert\vert_{\infty} \leq R_{00}
    \end{eqnarray*}
    Then the linear map $\Phi$ given in (\ref{map}) is a positive map on $B\left(\mathbf{C}^{d_1}\right)$ where $\vert\vert.\vert\vert_\infty$ stands for the operator norm.
\end{proposition}
 Let us recall that the H.W. observables are traceless operators and they form an orthogonal basis for Hermitian operators in $B\left(\mathbf{C}^{d}\right)$.
 Consider the linear map $\tilde{\Phi} : B\left(\mathbf{C}^d\right) \rightarrow B\left(\mathbf{C}^d\right)$ defined as 
\begin{equation} \label{H.W.1}
    \tilde{\Phi}_{k,l} (Y) =  Q_{k,l} Y~ Q_{k,l} + Q_{-k,-l} Y Q_{-k,-l}
\end{equation}
 Let us prove the following proposition first,

 \begin{proposition} \label{prop 7}
     The operators $Q_{m,n}$ where $m=0,1,...,d-1$ and $n=0,1,...,d-1$ are the eigenvectors of the map (\ref{H.W.1}) with the real eigenvalues
     \[2  \cos\left( \frac{ 2 \pi (kn-lm)}{d}\right) \].
      \end{proposition}

 \begin{proof}
     For any m and n, we shall prove the eigenvalue equation
     \[\tilde{\Phi}_{k,l} (Q_{m,n})= \lambda_{k,l,m,n} ~~Q_{m,n}\] with \[\lambda_{k,l,m,n}= 2  \cos\left( \frac{ 2 \pi (kn-lm)}{d}\right)\]
Now, \[\tilde{\Phi}_{k,l} (Q_{m,n})=Q_{k,l} Q_{m,n}~ Q_{k,l} + Q_{-k,-l} Q_{m,n} Q_{-k,-l} \]
From the definition of the H.W. observables we get, 

\[\tilde{\Phi}_{k,l} (Q_{m,n})=\left( \kappa D_{k,l}+\kappa^* D^{\dagger}_{k,l}\right)\left( \kappa D_{m,n}+\kappa^* D^{\dagger}_{m,n}\right)\left( \kappa D_{k,l}+\kappa^* D^{\dagger}_{k,l}\right)+\]\[\left( \kappa D_{-k,-l}+\kappa^* D^{\dagger}_{-k,-l}\right)\left( \kappa D_{m,n}+\kappa^* D^{\dagger}_{m,n}\right)\left( \kappa D_{-k,-l}+\kappa^* D^{\dagger}_{-k,-l}\right)\]

expanding the products and using the fact 
$D_{p,q}^{\dagger}=D_{-p,-q}$ we get,

\[\kappa^3 \left( D_{k,l}D_{m,n}D_{k,l}+D_{-k,-l}D_{m,n}D_{-k,-l}\right) + \kappa^{*3} \left( D_{-k,-l}D_{-m,-n}D_{-k,-l}+D_{k,l}D_{-m,-n}D_{k,l}\right) \]
\[+\kappa^2\kappa^* \left( 2 D_{-k,-l}D_{m,n}D_{k,l}+2 D_{k,l}D_{m,n}D_{-k,-l}+D_{k,l}D_{-m,-n}D_{k,l}+D_{-k,-l}D_{-m,-n}D_{-k,-l}\right)\]
\begin{eqnarray*}
   + \kappa \kappa^{*2}\left( 2 D_{-k,-l}D_{-m,-n}D_{k,l}+2 D_{k,l}D_{-m,-n}D_{-k,-l}+D_{-k,-l}D_{m,n}D_{-k,-l}+D_{k,l}D_{m,n}D_{k,l} \right)
\end{eqnarray*}
This further reduces to 
\begin{equation}\label{eigen}
    2 \kappa^2\kappa^* \left(  D_{-k,-l}D_{m,n}D_{k,l}+ D_{k,l}D_{m,n}D_{-k,-l}\right)+2 \kappa \kappa^{*2}\left(  D_{-k,-l}D_{-m,-n}D_{k,l}+ D_{k,l}D_{-m,-n}D_{-k,-l}\right)
\end{equation}
     since $\left(\kappa^3+\kappa \kappa^{*2} \right)=\left(\kappa^{*3}+\kappa^2 \kappa^{*} \right)=0$. Let us recall that \[D_{p,q}=exp\left(\frac{-i \pi p q}{d}\right) Z^p X^q ~~\text{and}~~Z^pX^q=exp\left(\frac{i 2 \pi p q}{d}\right) X^qZ^p\]
     Using the definition of $D_{p,q}$ and this commutation relation between powers of X and Z operators (\ref{eigen}) reduces to 
     \[\lambda_{k,l,m,n} \left(\kappa D_{m,n}+\kappa^* D_{-m,-n}\right)= \lambda_{k,l,m,n}~~ Q_{m,n}\]
     
 \end{proof}

 \begin{corollary}
     Let us consider a linear map $\Lambda : B\left(\mathbf{C}^d\right) \rightarrow B\left(\mathbf{C}^d\right)$ defined as 
\begin{equation}\label{H.W. map}
    \Lambda (Y) = \sum_{k,l=0}^{d-1} p_{k,l}~ Q_{k,l} Y~ Q_{k,l}
\end{equation}
   where d is an odd natural number then if  $p_{k,l}=p_{-k,-l}$, the matrix R corresponding to the map $\Lambda$ is diagonal.
 \end{corollary}
 \begin{proof}
         If d is an odd natural number and $p_{k,l}=p_{-k,-l}$ then we can rewrite the map (\ref{H.W. map}) as
         \begin{eqnarray*}
          \Lambda (Y) = p_{0,0} Q_{0,0} Y Q_{0,0}+\sum_{k,l}  p_{k,l} \left(Q_{k,l} Y Q_{k,l} +Q_{-k,-l} Y Q_{-k,-l}\right)
         \end{eqnarray*}
     Now using the proposition (\ref{prop 7}), we get
     \[\Lambda (Q_{m,n})=p_{0,0} Q_{m,n}+\sum_{k,l} p_{k,l}~ \lambda_{k,l,m,n} ~Q_{m,n}\]
     \[=\left( p_{0,0}+ \sum_{k,l} p_{k,l}~ \lambda_{k,l,m,n}\right) Q_{m,n}\]
     Since the set of H.W. observables $\lbrace Q_{m,n} \rbrace$ forms a basis, it shows from the eigenvalue equations that the matrix representation of the linear map $\Lambda$ with respect to the H.W. observable basis is diagonal and hence the corresponding R matrix is diagonal.
     \end{proof}

 Let us consider the dimension of the system be $d=3$ for a specific case study and for the brevity let us denote the H.W. observables as follows:
\[Q_{0,0}:=Q_0,Q_{0,1}:=Q_1,Q_{0,2}:=Q_2,Q_{1,0}:=Q_3,Q_{1,1}:=Q_4\]\[Q_{1,2}:=Q_5,Q_{2,0}:=Q_6,Q_{2,1}:=Q_7, Q_{2,2}:=Q_8\]

We are interested in positivity condition for the  linear map $\tilde\Lambda:B\left(\mathbf{C}^3 \right) \rightarrow B\left(\mathbf{C}^3\right) $ given by
\begin{equation}\label{d=3}
    \tilde\Lambda (Y)=\sum_{i=0}^8 p_i ~Q_i Y Q_i
\end{equation}
as a case study.
From proposition \ref{unital1} we have the following:
\begin{eqnarray*}
    p_1=p_2; p_3=p_6; p_4=p_8; p_5=p_7 ~\text{and}~ p_0=1-2(p_1+p_3+p_4+p_5)  ~~~~\text{imply}~~~~ \tilde\Lambda ~~\text{is unital}.
\end{eqnarray*}

Now by virtue of (\ref{map}) one can compute the R matrix corresponding to the map $\tilde\Lambda$ and finds that
\begin{eqnarray*}
    R_{00}=\sum_{i=0}^8 p_i
\end{eqnarray*}
\begin{eqnarray*}
    \text{vector}~~ t=s=\left( \frac{p_2-p_1}{2},\frac{p_1-p_2}{2},\frac{p_6-p_3}{2},\frac{p_4-p_8}{2},\frac{p_7-p_5}{2},\frac{p_3-p_6}{2},\frac{p_5-p_7}{2},\frac{p_4-p_8}{2}, \right)
\end{eqnarray*}
\\

Therefore from the proposition (\ref{unital1}) one has the following corollary.\\

\begin{corollary}
    The linear map defined in (\ref{d=3}) is unital if and only if
     \begin{eqnarray*}
    p_1=p_2; p_3=p_6; p_4=p_8; p_5=p_7 ~\text{and}~ p_0=1-2(p_1+p_3+p_4+p_5) 
    \end{eqnarray*}
\end{corollary}
\vspace{0.2cm}
Therefore the assumption of unitality reduces the number of independent parameters from 9 to 4 viz. $p_1,p_3,p_4,p_5$ say.\\

\begin{proposition}
    $\tilde\Lambda$ is unital completely positive if and only if 
    \begin{eqnarray*}
        p_0 \geq 0, p_1 \geq 0, p_3\geq 0 , p_4 \geq0, p_5 \geq0
    \end{eqnarray*}
    \text{Moreover} $p_0 \geq0$ implies $p_1+p_3+p_4+p_5 \leq \frac{1}{2}$
\end{proposition}
\vspace{0.2cm}

One would like to have positive but not completely positive maps out of this construction. In this regard one can note that to have positive but not completely positive maps not all of $p_0, p_1, p_3, p_4, p_5$ can be non-negative. From the sufficient condition for positivity we derive maximum how many of them can be non-negative for a unital map defined via \ref{d=3}\\

One finds that the assumption of unitality for the map defined via \ref{d=3} reduces the sub-matrix $\Delta$ in the R matrix to a diagonal matrix with following four terms in the diagonal each repeated twice.\\
\begin{equation}\label{eigenvalues}
    \lambda^{(1)}= 1-3 (p_3+p_4+p_5); \lambda^{(2)}= 1-3 (p_1+p_4+p_5);\lambda^{(3)}= 1-3 (p_1+p_3+p_5);\lambda^{(4)}= 1-3 (p_1+p_3+p_4);
\end{equation}
Now we prove the following theorem\\

\begin{theorem}
    Let us consider the unital linear map defined by (\ref{d=3}) with $2  \vert\vert\Delta\vert\vert_{\infty} \leq 1$ involving the independent parameters $p_1,p_3,p_4,p_5$ say. If $\lambda^{(i)}, ~i=1,2,3,4$ are of same sign, at least two of $p_1,p_3,p_4,p_5$ have to be positive.
\end{theorem}

\begin{proof}
    \textbf{Case 1: If $\lambda^{(i)} \geq 0 ~\text{for}~i=1,2,3,4$ } \\
    
    Given that the map defined by (\ref{d=3}) is unital one has the vector~~$t=\theta$. Moreover for this map the vector t and the vector s being the same, this map is also trace preserving with the vector $s=\theta$. Therefore  from proposition \ref{positive} one has if \begin{equation} \label{63}
        2  \vert\vert \Delta\vert\vert_{\infty} \leq 1
    \end{equation} then the map is positive. Note that $\vert\vert.\vert\vert_\infty$ is the operator norm which equals to the largest singular value of the operator. Therefore from (\ref{eigenvalues}) we get 
    \[2  \vert\vert\Delta\vert\vert_{\infty} \leq 1\]\[ \implies 2 max \lbrace \left( 1-3(p_3+p_4+p5)\right), \left( 1- 3 (p_1+p_4+p_5), \left( 1-3(p_1+p_3+p_5)\right) , \left(1-3(p_1+p_3+p_4)\right)\right)\rbrace \leq 1\]
    \[\implies 1+3 max \lbrace -(p_3+p_4+p5) , -(p_1+p_4+p_5), -(p_1+p_3+p_5), -(p_1+p_3+p_4) \rbrace \leq \frac{1}{2} \]
    \[\implies max \lbrace -(p_3+p_4+p5), -(p_1+p_4+p_5),-(p_1+p_3+p_5), - (p_1+p_3+p_4) \rbrace \leq -\frac{1}{6}\]
    \[\implies - min \lbrace (p_3+p_4+p5), (p_1+p_4+p_5), (p_1+p_3+p_5), (p_1+p_3+p_4)  \rbrace \leq - \frac{1}{6}\]
    \[\implies min \lbrace (p_3+p_4+p5), (p_1+p_4+p_5), (p_1+p_3+p_5), (p_1+p_3+p_4) \rbrace \geq \frac{1}{6}\]Therefore among $p_1, p_3, p_4,p_5$ at least two of them must be positive. \\

    \textbf{Case 2: If $\lambda^{(i)} \leq 0 ~\text{for}~i=1,2,3,4$ } \\

    As $\lambda^{(i)} \leq 0 ~i=1,2,3,4$  then we can see that all the terms 
    \begin{eqnarray*}
       (p_3+p_4+p5) \geq  \frac{1}{3}, (p_1+p_4+p_5) \geq  \frac{1}{3},   (p_1+p_3+p_5) \geq \frac{1}{3}, (p_1+p_3+p_4) \geq  \frac{1}{3}
    \end{eqnarray*}

    
    As  $\lambda^{(i)} \leq 0 ~i=1,2,3,4$ so the inequality  (\ref{63}) is always satisfied. For here we can conclude that atleast two of $p_1,p_3,p_4,p_5$ must be positive.
\end{proof}
\vspace{0.3cm}

\begin{remark}
    If one consider $p_0 = \frac{-1}{3}, p_1=p_3=p_4=p_5=\frac{1}{6}$, one gets 
    \begin{eqnarray*}
        \mathcal{R} (X)= \frac{1}{2} \left( Tr(X) \mathbf{I}_3 - X\right)
    \end{eqnarray*}
    which is the famous Reduction map on $d=3$.
\end{remark}

\vspace{1cm}

\section{Conclusion and future directions:} \label{V}
Heisenberg-Weyl (H.W.) operators provide a sort of generalization of Pauli operators in higher dimensions. They are Hermitian by construction and they form an operator basis in higher dimension. In this work we analyze some of the algebraic and spectral properties of such operators. If the dimension on which H.W. operators are defined be a prime number d, then they can be split into $d+1$ mutually exclusive subsets each consists of $d-1$ operators which commute with each other. This property resembles the commutation property of Weyl operators in prime dimensions.\\

Moreover we have analyzed linear maps arising from H.W. operators or H.W. observables. We derive a sufficient condition for unitality of the Heisenberg-Weyl (H.W.) maps i.e the maps that can be constructed with H.W. observables. This map is a generalization of Pauli maps defined on qubit systems. We note from the unitality condition that if the dimension of the system under consideration is an odd natural number, then the mutually commuting set of H.W. observables can further be divided into pairs. These pairs can be used to construct a linear map on the algebra of operators acting on odd dimensions which reduces to generalized Pauli channels for prime dimensions. We further deduce that the matrix representation of H.W. map with respect to H.W. basis is diagonal. Moreover the Reduction map can be obtained using H.W. maps.\\

Further it will be interesting to construct new examples of positive but not completely positive maps out of H.W. observables. Specifically construction of indecomposable positive maps will be a point of interest in this direction. Recently Schwarz condition for a class of qubit maps enjoying diagonal unitary and orthogonal symmetries has been derived in \cite{CB2024}. This allows one to characterize the Schwarz condition of qubit Pauli maps fully. It will be interesting to investigate the Schwarz condition for H.W. maps and to construct new examples of non completely positive Schwarz maps acting on high dimensional algebra of operators. 

\section{Acknowledgement:}

Authors would like to acknowledge Prof. Dariusz Chru\'sci\'nski for various fruitful discussions.  Saikat Patra is supported by \href{https://www.csir.res.in/}{{\sf CSIR, Govt. of India}} research fellowship file number 09/1184(0005)/2019-EMR-I.

\providecommand{\bysame}{\leavevmode\hbox to3em{\hrulefill}\thinspace}
\providecommand{\MR}{\relax\ifhmode\unskip\space\fi MR }
\providecommand{\MRhref}[2]{%
  \href{http://www.ams.org/mathscinet-getitem?mr=#1}{#2}
}
\providecommand{\href}[2]{#2}


\begin{thebibliography}{10}

\bibitem{AEHK16}
Ali Asadian, Paul Erker, Marcus Huber, and Claude Kl\"ockl,
  \emph{Heisenberg-weyl observables: Bloch vectors in phase space}, Phys. Rev.
  A \textbf{94} (2016), 010301.

\bibitem{breuer}
H.~P. Breuer and F.~Petruccione, \emph{The theory of open quantum systems},
  Oxford University Press, Great Clarendon Street, 2002.

\bibitem{breuerN}
Heinz-Peter Breuer, Elsi-Mari Laine, Jyrki Piilo, and Bassano Vacchini,
  \emph{Colloquium}, Rev. Mod. Phys. \textbf{88} (2016), 021002.

\bibitem{C75}
Man-Duen Choi, \emph{Positive semidefinite biquadratic forms}, Linear Algebra
  and its Applications \textbf{12} (1975), no.~2, 95--100.

\bibitem{C14}
Dariusz Chru\'{s}ci\'{n}ski, \emph{On kossakowski construction of positive maps
  on matrix algebras}, Open Systems \& Information Dynamics \textbf{21} (2014),
  no.~03, 1450001.

\bibitem{C2022}
Dariusz Chru{\'{s}}ci{\'{n}}ski, \emph{Dynamical maps beyond markovian regime},
  Physics Reports \textbf{992} (2022), 1--85.

\bibitem{CB2024}
Dariusz Chru{\'s}ci{\'n}ski and Bihalan Bhattacharya, \emph{A class of schwarz
  qubit maps with diagonal unitary and orthogonal symmetries}, Journal of
  Physics A: Mathematical and Theoretical \textbf{57} (2024), no.~39, 395202.

\bibitem{CS2014}
Dariusz Chru\'{s}ci\'{n}ski and Gniewomir Sarbicki, \emph{Entanglement
  witnesses: construction, analysis and classification}, Journal of Physics A:
  Mathematical and Theoretical \textbf{47} (2014), no.~48, 483001.

\bibitem{CS2016}
Dariusz Chru{\'{s}}ci{\'{n}}ski and Katarzyna Siudzi\ifmmode~\acute{n}\else
  \'{n}\fi{}ska, \emph{Generalized pauli channels and a class of non-markovian
  quantum evolution}, Phys. Rev. A \textbf{94} (2016), 022118.

\bibitem{alonso}
In\'es de~Vega and Daniel Alonso, \emph{Dynamics of non-markovian open quantum
  systems}, Rev. Mod. Phys. \textbf{89} (2017), 015001.

\bibitem{GT09}
Otfried G\"{u}hne and G\'{e}za T\'{o}th, \emph{Entanglement detection}, Physics
  Reports \textbf{474} (2009), no.~1, 1--75.

\bibitem{Ha98}
Kil-Chan Ha, \emph{Atomic positive linear maps in matrix algebras},
  Publications of the Research Institute for Mathematical Sciences \textbf{134}
  (1998).

\bibitem{Ha02}
Kil-Chan Ha, \emph{Positive projections onto spin factors}, Linear Algebra and
  its Applications \textbf{348} (2002), no.~1, 105--113.

\bibitem{Ha03}
\bysame, \emph{A class of atomic positive linear maps in matrix algebras},
  Linear Algebra and its Applications \textbf{359} (2003), no.~1, 277--290.

\bibitem{H96}
Michał Horodecki, Paweł Horodecki, and Ryszard Horodecki, \emph{Separability
  of mixed states: necessary and sufficient conditions}, Physics Letters A
  \textbf{223} (1996), no.~1, 1--8.

\bibitem{H97}
Pawel Horodecki, \emph{Separability criterion and inseparable mixed states with
  positive partial transposition}, Physics Letters A \textbf{232} (1997),
  no.~5, 333--339.

\bibitem{H09}
Ryszard Horodecki, Pawe\l{} Horodecki, Micha\l{} Horodecki, and Karol
  Horodecki, \emph{Quantum entanglement}, Rev. Mod. Phys. \textbf{81} (2009),
  865--942.

\bibitem{Jannesary2025}
Vahid Jannesary, Vahid Karimipour, and Dariusz Chru{\'s}ci{\'n}ski, \emph{A
  class of entanglement witnesses and a realignment-like criterion}, Scientific
  Reports \textbf{15} (2025), no.~1, 5718.

\bibitem{K03}
Seung-Hyeok Kye, \emph{Facial structures for unital positive linear maps in the
  two-dimensional matrix algebra}, Linear Algebra and its Applications
  \textbf{362} (2003), 57--73.

\bibitem{LACRS16}
Maciej Lewenstein, Remigiusz Augusiak, Dariusz Chru\'{s}ci\'{n}ski, Swapan
  Rana, and Jan Samsonowicz, \emph{Sufficient separability criteria and linear
  maps}, Phys. Rev. A \textbf{93} (2016), 042335.

\bibitem{M07}
W.A. Majewski, \emph{On positive decomposable maps}, Reports on Mathematical
  Physics \textbf{59} (2007), no.~3, 289--298.

\bibitem{MR17}
Marcin Marciniak and Adam Rutkowski, \emph{Merging of positive maps: A
  construction of various classes of positive maps on matrix algebras}, Linear
  Algebra and its Applications \textbf{529} (2017), 215--257.

\bibitem{NC10}
Michael~A. Nielsen and Isaac~L. Chuang, \emph{Quantum computation and quantum
  information: 10th anniversary edition}, Cambridge University Press, 2010.

\bibitem{OSAKA1991}
H.~Osaka, \emph{Indecomposable positive maps in low dimensional matrix
  algebras}, Linear Algebra and its Applications \textbf{153} (1991), 73--83.

\bibitem{Osaka93}
Hiroyuki Osaka, \emph{A series of absolutely indecomposable positive maps in
  matrix algebras}, Linear Algebra and its Applications \textbf{186} (1993),
  45--53.

\bibitem{book2}
Vern Paulsen, \emph{Completely bounded maps and operator algebras}, Cambridge
  Studies in Advanced Mathematics, Cambridge University Press, 2003.

\bibitem{P96}
Asher Peres, \emph{Separability criterion for density matrices}, Phys. Rev.
  Lett. \textbf{77} (1996), 1413--1415.

\bibitem{rivas1}
Angel Rivas, Susana~F Huelga, and Martin~B Plenio, \emph{Quantum
  non-markovianity: characterization, quantification and detection}, Reports on
  Progress in Physics \textbf{77} (2014), no.~9, 094001.

\bibitem{Stinespring55}
Forrest~W. Stinespring, \emph{Positive functions on c * -algebras}, Proc. Amer.
  Math. Soc. \textbf{6} (1955), 211--216.

\bibitem{book1}
E.~Stormer, \emph{Positive linear maps of operator algebras}, Springer
  Monographs in Physics; https://www.springer.com/gp/book/9783642343681, 2013.

\bibitem{S63}
Erling St{\o}rmer, \emph{Positive linear maps of operator algebras}, Acta
  Mathematica \textbf{110} (1963), no.~1, 233--278.

\bibitem{Tang86}
Wai-Shing Tang, \emph{On positive linear maps between matrix algebras}, Linear
  Algebra and its Applications \textbf{79} (1986), 33--44.

\bibitem{W18}
John Watrous, \emph{The theory of quantum information}, Cambridge University
  Press, 2018.

\bibitem{W13}
Mark~M. Wilde, \emph{Quantum information theory}, Cambridge University Press,
  2013.

\bibitem{WORONO76}
S.L. Woronowicz, \emph{Positive maps of low dimensional matrix algebras},
  Reports on Mathematical Physics \textbf{10} (1976), no.~2, 165--183.

\bibitem{ZC13}
Justyna~P. Zwolak and Dariusz Chru\'{s}ci\'{n}ski, \emph{New tools for
  investigating positive maps in matrix algebras}, Reports on Mathematical
  Physics \textbf{71} (2013), no.~2, 163--175.

\end{thebibliography}

\end{document}